\documentclass[11pt]{article}

\usepackage{amsmath,amssymb,amsthm}

\usepackage[pdfpagemode=UseOutlines,
pdfstartview={XYZ null null 1},
pdfpagelayout=SinglePage,pdfhighlight=/P,
colorlinks=true,linkcolor=black,citecolor=blue,urlcolor=blue]{hyperref}

\advance\textwidth 30mm  \advance\oddsidemargin -15mm
\advance\textheight 25mm \advance\topmargin -15mm
\allowdisplaybreaks[0]

\DeclareMathOperator{\const}{const}

\theoremstyle{plain}
\newtheorem{proposition}{Proposition}
\theoremstyle{remark}
\newtheorem{remark}{Remark}

\title{Painlev\'e type reductions for the non-Abelian Volterra lattices}
\date{18 November 2020}
\author{V.E. Adler\thanks{L.D.~Landau Institute for Theoretical Physics, Chernogolovka, Russian Federation. E-mail:~adler@itp.ac.ru}}

\newcommand{\tp}{^{\mbox{\rm\tiny T}}}

\begin{document}
\maketitle

\rightline{\em To the memory of A.B.\,Shabat and R.I.\,Yamilov}
\medskip

\begin{abstract}
The Volterra lattice admits two non-Abelian analogs that preserve the integrability property. For each of them, the stationary equation for non-autonomous symmetries defines a constraint that is consistent with the lattice and leads to Painlev\'e-type equations. In the case of symmetries of low order, including the scaling and master-symmetry, this constraint can be reduced to second order equations. This gives rise to two non-Abelian generalizations for the discrete Painlev\'e equations dP$_1$ and dP$_{34}$ and for the continuous Painlev\'e equations P$_3$, P$_4$ and P$_5$. 
\medskip

\noindent Keywords: non-Abelian system, symmetry, constraint, Volterra lattice, Painlev\'e equation, isomonodromic deformation, quasi-determinant.
\end{abstract}

\section{Introduction}

We study equations of Painlev\'e type that arise as symmetry constraints for the non-Abelian Volterra lattices
\begin{equation}\label{VL1}\tag*{\mbox{VL$^1$}}
 u_{n,x}=u_{n+1}u_n-u_nu_{n-1}
\end{equation}
and
\begin{equation}\label{VL2}\tag*{\mbox{VL$^2$}}
 u_{n,x}=u\tp_{n+1}u_n-u_nu\tp_{n-1}
\end{equation}
where $u$ and $u\tp$ can be viewed as a matrix of any fixed size and its transpose. Up to the author's knowledge, equation \ref{VL1} appeared in \cite{Salle_1982} and equation \ref{VL2} is new. There exists a sequence of substitutions that links these two equations, as described in Section \ref{s:mod}, but this is an implicit transformation and we prefer to treat both equations independently, albeit parallel to each other.

Both equations VL$^{1,2}$ are integrable in the sense of existence of infinite hierarchy of higher symmetries and conservation laws. In this paper, these properties are almost never used, since we consider only low-dimensional reductions associated with symmetries containing $u_{n+k}$ for $|k|\le2$. These symmetries are listed in Section \ref{s:sym} and include: the lattice itself, its simplest higher symmetry, the classical scaling symmetry and the master-symmetry. The stationary equation for a linear combination of these flows determines a constraint which is consistent with the lattice (there are several non-equivalent cases depending on the choice of coefficients of the linear combination). This constraint admits a reduction of the order; as a result, a discrete equation of the Painlev\'e type arises, while the dynamics in $x$ is reduced to a continuous Painlev\'e equation. The theory of non-Abelian Painlev\'e equations appeared not very long ago (we mention \cite{Balandin_Sokolov_1998} as one of the first publications on this topic). Now this area is actively developing and includes many studied examples, but there are still many blank spots and a classification of non-Abelian case is lacking.

Section \ref{s:zcr} contains zero curvature representations for the flows we need; under the constraint, they turn into the isomonodromic Lax pairs for the Painlev\'e equations.

In the scalar case $u_n\in\mathbb C$, the first example of a constraint of this type, leading to the dP$_1$ and P$_4$ equations, appeared in \cite{Its_Kitaev_Fokas_1990, Fokas_Its_Kitaev_1991}, see also \cite{Grammaticos_Ramani_1998, Grammaticos_Ramani_2014}. In this example the master-symmetry is not used and the non-isospectrality of the reduction is due to the scaling symmetry. The non-Abelian version of this reduction for both equations VL$^{1,2}$ is given in Section \ref{s:constr1}. The calculations here are quite straightforward and lead to two non-commutative analogs of dP$_1$. Notice that these analogs do not coincide with equations studied in \cite{Gordoa_Pickering_Zhu_2013, Cassatella-Contra_Manas_2012, Cassatella-Contra_Manas_Tempesta_2018} in the framework of other approaches. The respective analogs of P$_4$ equation for $y=u_n$ can be written as
\[
 y_{xx}=\frac{1}{2}y_xy^{-1}y_x+\frac{1}{2}[ky-\gamma y^{-1},y_x]+\frac{3}{2}y^3+4xy^2+2(x^2-\alpha)y-\frac{\gamma^2}{2}y^{-1},
\]
where the coefficient $k$ takes two values: 1 for \ref{VL1} and $-3$ for \ref{VL2}.

Section \ref{s:constr2} is devoted to more difficult constraints related with master-symmetries of VL$^{1,2}$. In the scalar case, this reduction was studied in \cite{Adler_Shabat_2019} where it was shown that it leads to equations dP$_{34}$ and P$_5$ or P$_3$ (depending on the values of parameters). Now we present non-Abelian generalizations for these equations. However, the continuous part of the answer is left in the form of a system of two first-order ODEs, since reduction to one second-order ODE leads to rather complicated formulas.

An essential difference in the non-Abelian case is related with the procedure of reduction of order. We start from a stationary equation which is some 5-point O$\Delta$E:
\[
 f_n(u_{n-2},u_{n-1},u_n,u_{n+1},u_{n+2};x,\mu,\nu)=0,
\]
where $f_n$ is a polynomial in $u_k$ and $u\tp_k$ with coefficients depending on $n,x$ and scalar parameters $\mu,\nu$, and our final result is a discrete Painlev\'e equation of the form
\[
 g_n(u_{n-1},u_n,u_{n+1};x,\mu,\nu,\varepsilon,\delta)=0
\]
with additional constants $\varepsilon,\delta\in\mathbb C$. In the scalar case, both equations are completely equivalent since the reduction of order is made due to the first integrals which enter the equation as parameters. In the non-Abelian case, we use {\em partial} first integrals instead. As a result, the obtained Painlev\'e equation defines only a subclass of special solutions of the original equation (two of four matrix initial data are replaced by values of scalar parameters).
\smallskip

{\em Notations.} We use Greek letters to denote scalar coefficients, that is, elements of $\mathbb C$. By default, Latin letters stand for elements of an associative algebra ${\cal A}$ over $\mathbb{C}$, with the exception of the independent variables $n,x$ and $t_i$ (used for the flows from the lattice hierarchies) which, of course, are scalars. The involution $u\tp$ on $\cal A$ is such that $(u\tp)\tp=u$, $(\alpha u+\beta v)\tp=\alpha u\tp+\beta v\tp$ and $(uv)\tp=v\tp u\tp$ for any $u,v\in \cal A$ and $\alpha,\beta\in\mathbb C$. We also assume that the algebra ${\cal A}$ has the unity 1 and that expressions involving inverse elements are allowed. In expressions like $u_n+\alpha$, the term $\alpha$ is understood as the scalar $\alpha$ multiplied by $1\in\cal A$, which does not lead to a confusion. For clarity, we can think of ${\cal A}$ as the algebra of matrices of arbitrary fixed size with the operations of matrix transpose and taking the inverse (it is also possible to use the conjugate transpose $u^*$, assuming that some part of scalars should be real).

\section{Modified equations}\label{s:mod}

Although \ref{VL1} and \ref{VL2} equations are very similar, there is no an explicit invertible change between them. It is clear that it is possible to get rid of the involution in \ref{VL2} by passing to the sequence of variables $\dots,u\tp_{n-1},u_n,u\tp_{n+1},u_{n+2},\dots$, on the expense that equations for variables with even and odd numbers become different. Then \ref{VL1} takes the form of two-component system of Toda type
\[
 p_{n,x}=q_np_n-p_nq_{n-1},~~ q_{n,x}=p_{n+1}q_n-q_np_n\quad (u_{2n}=p_n,~ u_{2n+1}=q_n),
\]
and \ref{VL2} turns into another system
\[
 p_{n,x}=q_np_n-p_nq_{n-1},~~ q_{n,x}=q_np_{n+1}-p_nq_n\quad (u_{2n}=p_n,~ u\tp_{2n+1}=q_n).
\]
Nevertheless, both lattices are related indeed, but in more complicated and unexpected way via the sequence of difference substitutions
\[
 \ref*{VL1}~ \leftarrow~ \mbox{mVL}^1~ \leftarrow~ \mbox{pot-mVL}~ \to~ \mbox{mVL}^2~ \to~ \ref*{VL2},
\]
which act in opposite directions and involve inverse elements in ${\cal A}$. The equations participating in this sequence are
\begin{align*}
 \mbox{VL$^1$}:&\quad u_{n,x}= u_{n+1}u_n-u_nu_{n-1},\\
\mbox{mVL$^1$}:&\quad v_{n,x}= v_{n+1}(v^2_n-\alpha^2)-(v^2_n-\alpha^2)v_{n-1},\\
\mbox{pot-mVL}:&\quad w_{n,x}= (w_{n+1}+2\alpha w_n)(w^{-1}_{n-1}w_n+2\alpha),\\
\mbox{mVL$^2$}:&\quad v_{n,x}= (v_n-\alpha)v_{n+1}(v_n+\alpha)-(v_n+\alpha)v_{n-1}(v_n-\alpha),\\
 \mbox{VL$^2$}:&\quad u_{n,x}= u\tp_{n+1}u_n-u_nu\tp_{n-1}.
\end{align*}
The following substitutions are easily verified by direct computation. The VL$^1$ is related with the modified Volterra lattice mVL$^1$ by the well-known discrete Miura map \cite{Salle_1982, Bogoyavlensky_1991a, Casati_Wang_2020} 
\[
 \mbox{VL}^1\leftarrow\mbox{mVL}^1:~~ u_n=(v_{n+1}+\alpha)(v_n-\alpha).
\]
The substitution from the potential lattice pot-mVL to mVL$^1$ is of the form
\[
 \mbox{mVL}^1\leftarrow\mbox{pot-mVL}:~~ v_n=w_{n+1}w^{-1}_n+\alpha,
\]
and the substitution to mVL$^2$ differs only in the order of the factors:
\[
 \mbox{pot-mVL}\to\mbox{mVL}^2:~~ v_n=w^{-1}_nw_{n+1}+\alpha.
\]
Finally, the last substitution is
\[
 \mbox{mVL}^2\to\ref{VL2}:~~ u_{2n}=(v_{2n}+\alpha)(v_{2n-1}+\alpha),~~ 
 u_{2n+1}=(v\tp_{2n+1}-\alpha)(v\tp_{2n}-\alpha).
\]
The mVL$^2$ equation was studied (for zero parameter $\alpha$) in \cite{Adler_Svinolupov_Yamilov_1999}, where it was stated that it did not admit an analog of the Miura map. We see that this is not so: the trick is that the transformation should be defined by different formulas for even and odd variables and using the involution $\tp$, although the mVL$^2$ equation itself does not contain it.

In principle, the above substitutions can be applied to the reductions of equations \ref{VL1} and \ref{VL2} which we study below. However, this is not easy given the implicit nature of this transformation. Therefore, we leave the problem of studying these relations for further research, and in this work we consider both lattices and the Painlev\'e equations obtained from them independently, although parallel to each other.

\begin{remark}
It is interesting to note the incomplete analogy with the sequence of substitutions
\[ 
 \mbox{KdV}\xleftarrow{u=v^2\pm v_x+\alpha}
 \mbox{mKdV}^1\xleftarrow{v= w_xw^{-1}}
 \mbox{pot-mKdV}\xrightarrow{v= w^{-1}w_x}
 \mbox{mKdV}^2
\]
(see eg. \cite{Kupershmidt_2000}) for the equations
\begin{align*}
     \mbox{KdV}:&\quad u_t= u_{xxx}-3uu_x-3u_xu,\\
\mbox{mKdV$^1$}:&\quad v_t= v_{xxx}-3v^2v_x-3v_xv^2-6\alpha v_x,\\
\mbox{pot-mKdV}:&\quad w_t= w_{xxx}-3w_{xx}w^{-1}w_x-6\alpha w_x,\\
\mbox{mKdV$^2$}:&\quad v_t= v_{xxx}+3[v,v_{xx}]-6vv_xv-6\alpha v_x.
\end{align*}
These equations can be obtained from the corresponding lattice equations by continuous limit, but no continuous analog of \ref{VL2} is known.
\end{remark}

\section{Symmetries}\label{s:sym}

For each lattice \ref{VL1} or \ref{VL2}, the evolution derivation $D_x=D_{t_1}$ is a member of infinite-dimensional Lie algebra with generators satisfying commutation relations
\begin{equation}\label{DD}
 [D_{t_i},D_{t_j}]=0,~~ [D_{\tau_i},D_{t_j}]=jD_{t_{j+i-1}},~~ [D_{\tau_i},D_{\tau_j}]=(j-i)D_{\tau_{j+i-1}},
 \quad i,j\ge1.
\end{equation}
This implies, in particular, that all derivations $D_{t_i}$ and $xD_{t_i}+D_{\tau_i}$ commute with $D_x$ and that the stationary equation for any linear combination of these derivations defines a constraint consistent with the lattice. The Painlev\'e equations or their higher analogs appear if such a constraint involves at least one derivation of the form $xD_{t_i}+D_{\tau_i}$. In this paper, we consider only low order constraints obtained by use of derivations $D_x$, $D_{t_2}$, $D_{\tau_1}$ and $D_{\tau_2}$ which we now write down explicitly. The relations (\ref{DD}) for these flows can be verified directly.

For \ref{VL1}, the symmetry $D_{t_2}$ is defined by equation
\begin{equation}\label{VL1.t2}
 u_{n,t_2}=(u_{n+2}u_{n+1}+u^2_{n+1}+u_{n+1}u_n)u_n-u_n(u_nu_{n-1}+u^2_{n-1}+u_{n-1}u_{n-2})
\end{equation}
and its analog for \ref{VL2} reads
\begin{equation}\label{VL2.t2}
 u_{n,t_2}= (u\tp_{n+1}u_{n+2} +(u\tp_{n+1})^2 +u_nu\tp_{n+1})u_n
  -u_n(u\tp_{n-1}u_n +(u\tp_{n-1})^2 +u_{n-2}u\tp_{n-1}).
\end{equation}
The classical scaling symmetry $D_{\tau_1}$ for both lattices is of the form
\[
 u_{n,\tau_1}=u_n.
\]
The master-symmetry $D_{\tau_2}$ for \ref{VL1} involves an additional nonlocal variable $s_n$:
\begin{equation}\label{VL1.tau2}
 u_{n,\tau_2} = \bigl(n+\tfrac{3}{2}\bigr)u_{n+1}u_n+u^2_n-\bigl(n-\tfrac{3}{2}\bigr)u_nu_{n-1}+[s_n,u_n],\quad s_n-s_{n-1}=u_n.
\end{equation}
Notice, that the derivatives of $s_n$ with respect to $x$ and $t_2$ are local, for instance $s_{n,x}=u_{n+1}u_n$. It is clear that in the scalar case $u_n\in\mathbb C$ introducing of this nonlocality is not necessary (a description of local master-symmetries for the Volterra type lattices can be found, e.g. in \cite{Cherdantsev_Yamilov_1995, Adler_Shabat_Yamilov_2000}). As it turns out, for the lattice \ref{VL2} the nonlocality is not necessary as well: the master-symmetry in this case is of the form
\begin{equation}\label{VL2.tau2}
 u_{n,\tau_2} = \bigl(n+\tfrac{3}{2}\bigr)u\tp_{n+1}u_n+u^2_n-\bigl(n-\tfrac{3}{2}\bigr)u_nu\tp_{n-1}.
\end{equation}
The relations (\ref{DD}) imply that the derivations of the form
\begin{equation}\label{Dt}
 D_t=\mu_1(xD_{t_2}+D_{\tau_2})+\mu_2(xD_x+D_{\tau_1})+\mu_3D_{t_2}+\mu_4D_x
\end{equation}
commute with $D_x$. Hence it follows that the stationary equation $u_{n,t}=0$ is a constraint compatible with the lattice. Taking into account the scaling transformation and the shift $x\to x-x_0$, one can see that there are two nonequivalent cases leading to the Painlev\'e type equations:
\begin{equation}\label{constr1}
 u_{n,t_2}+2(xu_{n,x}+u_{n,\tau_1})=0
\end{equation}
and
\begin{equation}\label{constr2}
 xu_{n,t_2}+u_{n,\tau_2}-2\mu(xu_{n,x}+u_{n,\tau_1})-\nu u_{n,x}=0.
\end{equation}
These are the constraints that will be the subject of our further study.

\begin{remark}
In addition to the above symmetries, there are symmetries of the form $u_{n,y}=[a,u_n]$ which are specific to the non-Abelian case. For the lattice \ref{VL1}, $a$ can be any constant element of $\cal A$, and for \ref{VL2} it should satisfy $a=-a\tp$. Introducing such term into (\ref{Dt}) may be interesting because it violates the $GL$-invariance, but we do not consider this possibility in this paper.
\end{remark}

\begin{remark}
By use of \ref{VL1} or \ref{VL2} equation, one can find all variables $u_{n+k}$ as expressions of variables $u_n,u_{n+1}$ and their derivatives with respect to $x$. Therefore, any symmetry of the lattice can be rewritten as some coupled PDE system. In particular, for the \ref{VL1} equation and its symmetry (\ref{VL1.t2}), the pair $(p,q)=(u_n,u_{n+1})$ satisfies, for any $n$, the associated system
\[
 \left\{\begin{aligned}
 & q_{t_2}=~ q_{xx}+2q_xq+2(qp)_x+2[qp,q],\\
 & p_{t_2}= -p_{xx}+2pp_x+2(qp)_x+2[qp,p].
 \end{aligned}\right.
\]
Similarly, for \ref{VL2} and its symmetry (\ref{VL2.t2}), the pair $(p,q)=(u_n,u\tp_{n+1})$ satisfies the system
\[
 \left\{\begin{aligned}
 & q_{t_2}=~ q_{xx}+2q_xq+2(pq)_x+2[pq,q],\\
 & p_{t_2}= -p_{xx}+2p_xp+2(qp)_x+2[p,qp].
 \end{aligned}\right.
\]
These systems were found in \cite{Adler_Sokolov_2020} when classifying non-Abelian systems with higher conservation laws. They give two generalizations of the Levi scalar system from \cite{Levi_1981}, where the concept of the associated system was introduced; see also \cite{Shabat_Yamilov_1991} for a systematic study of relations between (scalar) Volterra and Toda-type lattices and nonlinear Schr\"odinger type equations.
\end{remark}

\section{Zero curvature representations}\label{s:zcr}

{\em \ref{VL1} equation.} It is easy to check that it is equivalent to the matrix equation
\begin{equation}\label{VL1.Lx}
 L_{n,x}=U_{n+1}L_n-L_nU_n
\end{equation}
with
\begin{equation}\label{VL1.LU}
 L_n=\begin{pmatrix}
    \lambda & \lambda u_n \\
    -1 & 0
 \end{pmatrix},\quad
 U_n=\begin{pmatrix}
  \lambda+u_n & \lambda u_n \\
  -1 & u_{n-1}
 \end{pmatrix}.
\end{equation}
The entire \ref{VL1} hierarchy is associated with linear equations of the general form 
\[
 \Psi_{n+1}=L_n\Psi_n,\quad \Psi_{n,t}+\kappa(\lambda)\Psi_{n,\lambda}=V_n\Psi_n
\]
with the above matrix $L_n$. The compatibility condition for these equations is written as the zero curvature representation
\begin{equation}\label{VL1.Lt}
 L_{n,t}+\kappa(\lambda)L_{n,\lambda}=V_{n+1}L_n-L_nV_n.
\end{equation}
Any member of the hierarchy is equivalent to (\ref{VL1.Lt}) for suitable choice of $V_n=V^{(t)}_n$, with factors $\kappa^{(t_j)}=0$ for the derivations $D_{t_j}$ and $\kappa^{(\tau_j)}=\lambda^j$ for $D_{\tau_j}$. It is easy to prove that the matrix $V_n$ has the structure
\[
 V_n=\begin{pmatrix}
   a_n & -\lambda c_{n+1}u_n \\
   c_n & \lambda c_n+a_{n-1}
  \end{pmatrix}.
\]  
In particular, the choice $a_n=\lambda+u_n$ and $c_n=-1$ brings to the matrix $U_n=V^{(x)}_n$ from (\ref{VL1.LU}). The matrices for other flows we need are obtained with the following choice:
\begin{align*}
 V^{(t_2)}_n:&~~ a_n=\lambda^2+\lambda u_n+u_{n+1}u_n+u^2_n+u_nu_{n-1},~~ c_n= -\lambda-u_n-u_{n-1},\\
 V^{(\tau_1)}_n:&~~ a_n=n+1,~~ c_n=0,\\
 V^{(\tau_2)}_n:&~~ a_n= (n-\tfrac{1}{2})(\lambda+u_n)+s_n,~~ c_n= -n+\tfrac{1}{2}.
\end{align*}
Notice, that $V^{(\tau_2)}_n$ contains nonlocality $s_n$ even in the scalar case, when the master-symmetry (\ref{VL1.tau2}) becomes local.\medskip

{\em \ref{VL2} equation.} In this case, the zero curvature representation has a less familiar form
\begin{equation}\label{VL2.Lx}
 L_{n,x}=U_{n+1}L_n+L_nU\tp_n,
\end{equation}
with the matrices
\begin{equation}\label{VL2.LU}
 L_n=\begin{pmatrix}
   1 & -\lambda\\
   0 & \lambda u_n
  \end{pmatrix},\quad
 U_n=\begin{pmatrix}
  \tfrac{1}{2}\lambda & 1\\
  -\lambda u_{n-1} & -\tfrac{1}{2}\lambda-u_{n-1}+u\tp_n
 \end{pmatrix},
\end{equation}
where $U\tp$ denotes the result of the matrix transpose of $U$ applying the involution to each element (that is, if $\cal A$ is the matrix algebra then this is just the block transpose). Representations of this type arise if the linear equations for $\psi$-functions are given differently for even and odd numbers (cf. with the alternating Miura map from Section \ref{s:mod}):
\begin{gather*}
 \Psi_{2n+1}=L_{2n}\Psi_{2n}=L\tp_{2n+1}\Psi_{2n+2},\\
 \Psi_{2n,x}=-U\tp_{2n}\Psi_{2n},\quad \Psi_{2n+1,x}=U_{2n+1}\Psi_{2n+1}.
\end{gather*}
It is easy to see that the compatibility conditions for these equations have the same form (\ref{VL2.Lx}) regardless of the parity of $n$. Note that from here we can pass to the representation of the form $M_{2n,x}=A_{2n+2}M_{2n}-M_{2n}A_{2n}$ with matrices $M_{2n}=(L\tp_{2n+1})^{-1}L_{2n}$. Representations of this type (with the step by two lattice sites) were used for some equations of Volterra type also in the scalar situation \cite{Garifullin_Mikhailov_Yamilov_2014}. 

More generally, all derivations from the \ref{VL2} hierarchy correspond to the linear equations
\[
 \Psi_{2n,t}+\kappa(\lambda)\Psi_{2n,\lambda}=-V\tp_{2n}\Psi_{2n},\quad 
 \Psi_{2n+1,t}+\kappa(\lambda)\Psi_{2n+1,\lambda}=V_{2n+1}\Psi_{2n+1},
\]
with the same matrix $L_n$, which leads to representations of the form
\begin{equation}\label{VL2.Lt}
 L_{n,t}+\kappa(\lambda)L_{n,\lambda}=V_{n+1}L_n+L_nV\tp_n,
\end{equation}
with $\kappa^{(t_j)}=0$ and $\kappa^{(\tau_j)}=\lambda^j$ as before. In particular, we are interested in the derivations $D_{t_2}$, $D_{\tau_1}$ and $D_{\tau_2}$ which correspond to the matrices
\[
 V^{(t_2)}_n=\begin{pmatrix}
  \tfrac{1}{2}\lambda^2+\lambda u_{n-1} & \lambda+a\tp_n\\
  -\lambda u_{n-1}(\lambda+a_{n-1}) & 
  -\tfrac{1}{2}\lambda^2-u_{n-1}(\lambda+a_{n-1})+u\tp_na_{n+1} 
  \end{pmatrix},
\]
where $a_n=u_n+u\tp_{n-1}$, and
\[
 V^{(\tau_1)}_n=\begin{pmatrix}
  0 & 0\\
  0 & 1 
 \end{pmatrix},\quad
 V^{(\tau_2)}_n=nU_n+\frac{1}{2}\begin{pmatrix}
  -\lambda & -1\\
  3\lambda u_{n-1} & 2\lambda+3u_{n-1}+u\tp_n
 \end{pmatrix}.
\]

{\em Transition to the isomonodromic Lax pairs.} The matrix $V_n$ and the factor $\kappa(\lambda)$ corresponding to the linear combination (\ref{Dt}) are constructed from matrices and factors corresponding to the basic derivations:
\[
 V_n=\mu_1(xV^{(t_2)}+V^{(\tau_2)}_n)+\mu_2(xU_n+V^{(\tau_1)}_n)+\mu_3V^{(t_2)}_n+\mu_4U_n,\quad
 \kappa=\mu_1\lambda^2+\mu_2\lambda.
\] 
This matrix satisfies equation (\ref{VL1.Lt}) or (\ref{VL2.Lt}), depending on which hierarchy we are considering. In addition, in both cases the compatibility condition for the derivations $D_x$ and $D_t$ is satisfied:
\begin{equation}\label{Ut}
 U_{n,t}+\kappa(\lambda)U_{n,\lambda}-V_{n,x}=[V_n,U_n].
\end{equation}
On passing to the stationary equation for (\ref{Dt}), the derivative with respect to $t$ disappears and the zero curvature representations (\ref{VL1.Lt}), (\ref{VL2.Lt}) and (\ref{Ut}) turn into isomonodromic Lax pairs. The discrete part of the constraint is equivalent to an equation of the form
\begin{equation}\label{VL1.Llambda}
 \kappa(\lambda)L_{n,\lambda}=V_{n+1}L_n-L_nV_n
\end{equation}
or of the form
\begin{equation}\label{VL2.Llambda}
 \kappa(\lambda)L_{n,\lambda}=V_{n+1}L_n+L_nV\tp_n,
\end{equation}
and the continuous part is equivalent to an equation of the form
\begin{equation}\label{Ulambda}
 \kappa(\lambda)U_{n,\lambda}-V_{n,x}=[V_n,U_n]
\end{equation}
with the matrix entries simplified modulo the constraint under study.

\section{Non-Abelian analogs of \texorpdfstring{dP$_1$}{dP1} and \texorpdfstring{P$_4$}{P4} equations}\label{s:constr1}

For the \ref{VL1} equation, the constraint (\ref{constr1}) takes the form of difference equation
\begin{multline*}
 u_{n+2}u_{n+1}u_n+u^2_{n+1}u_n+u_{n+1}u^2_n-u^2_nu_{n-1}-u_nu^2_{n-1}-u_nu_{n-1}u_{n-2}\\
  +2x(u_{n+1}u_n-u_nu_{n-1})+2u_n=0,\quad
\end{multline*}
and for the \ref{VL2} it reads
\begin{multline*}
 u\tp_{n+1}u_{n+2}u_n +(u\tp_{n+1})^2u_n +u_nu\tp_{n+1}u_n
   -u_nu\tp_{n-1}u_n -u_n(u\tp_{n-1})^2 -u_nu_{n-2}u\tp_{n-1} \\
 +2x(u\tp_{n+1}u_n-u_nu\tp_{n-1})+2u_n=0.\quad
\end{multline*}
Each of these equations can be represented as $F_{n+1}u_n-u_nF_{n-1}=0$, and it turns out that the equality $F_n=0$ is itself a constraint consistent with the lattice equation. This brings to the following non-Abelian analogs of dP$_1$ (we remark that different non-Abelian versions of dP$_1$ were studied in \cite{Cassatella-Contra_Manas_2012, Cassatella-Contra_Manas_Tempesta_2018}):
\begin{gather}
\label{dP11}\tag*{\mbox{dP$_1^1$}}
 u_{n+1}u_n+u^2_n+u_nu_{n-1}+2xu_n+n-\nu+(-1)^n\varepsilon=0,\\
\label{dP12}\tag*{\mbox{dP$_1^2$}}
 u\tp_{n+1}u_n+u^2_n+u_nu\tp_{n-1}+2xu_n+n-\nu+(-1)^n\varepsilon=0,
\end{gather}
where scalars $\nu$ and $\varepsilon$ play the role of integration constants. (This is only a partial first integral. The integration constants must belong to the center of $\cal A$, because they have to commute with $u_n$ which is a general element of $\cal A$.) Moreover, for consistency with the corresponding lattice, these scalars must be independent of $x$.

\begin{proposition}
The VL$^i$ equation is consistent with dP$_1^i$, $i=1,2$, for arbitrary constants $\nu,\varepsilon\in\mathbb C$.
\end{proposition}
\begin{proof} 
This is proved by direct calculation as follows. Let $F_n$ be the left-hand side of \ref{dP11}. By differentiating this expression in virtue of \ref{VL1}, we obtain the identity
\[
 F_{n,x}=(F_{n+1}-F_n)u_n+u_n(F_n-F_{n-1}).
\]
Similarly, if $F_n$ is the left-hand side of \ref{dP12} then differentiating in virtue of \ref{VL2} gives
\[
 F_{n,x}=(F\tp_{n+1}+F_n)u_n-u_n(F_n+F\tp_{n-1}).
\]
In both cases, the derivative of the constraint $F_n=0$ vanishes identically due to the constraint itself, as required.
\end{proof}

The obtained constraint turns the lattice equation into a closed system of two first order equations for the variables $u_{n-1}$ and $u_n$. It is possible to rewrite this system as a second-order ODE which generalize P$_4$. 

\begin{proposition}
If a solution of the VL$^i$ equation satisfies the constraint P$_4^i$ then any its component $y=u_n$ satisfies the P$_4^i$ equation, $i=1,2\!:$ 
\begin{align}
\label{P41}\tag*{\mbox{P$_4^1$}}
 &y_{xx}=\frac{1}{2}y_xy^{-1}y_x+\frac{1}{2}[y-\gamma y^{-1},y_x]+\frac{3}{2}y^3+4xy^2+2(x^2-\alpha)y-\frac{\gamma^2}{2}y^{-1},\\
\label{P42}\tag*{\mbox{P$_4^2$}}
 &y_{xx}=\frac{1}{2}y_xy^{-1}y_x-\frac{1}{2}[3y+\gamma y^{-1},y_x]+\frac{3}{2}y^3+4xy^2+2(x^2-\alpha)y-\frac{\gamma^2}{2}y^{-1}
\end{align}
with the values of parameters
\[
 \alpha=\gamma_{n-1}-\gamma_n/2+1,\quad \gamma=\gamma_n:=n-\nu+(-1)^n\varepsilon.
\]
\end{proposition}
\begin{proof}
In the \ref{VL1} case, we have the following system for the variables $(p,y)=(u_{n-1},u_n)$:
\begin{equation}\label{P41.pq}
 p_x= 2yp+p^2+2xp+\gamma_{n-1},\quad y_x=-y^2-2yp-2xy-\gamma_n.
\end{equation}
The second equation implies
\begin{equation}\label{py}
 p=x-\frac{1}{2}(y^{-1}y_x+y+\gamma_ny^{-1})
\end{equation}
and substitution into the first equation gives \ref{P41}. Similarly, in the \ref{VL2} case, the variables $(p,y)=(u\tp_{n-1},u_n)$ satisfy the system
\begin{equation}\label{P42.pq}
 p_x= 2py+p^2+2xp+\gamma_{n-1},\quad y_x=-y^2-2yp-2xy-\gamma_n
\end{equation}
which differs from the previous one just by one term. The elimination of $p$ brings to \ref{P42}.
\end{proof}

Notice, that the substitution (\ref{py}) (accompanied by additional involution $p\to p\tp$ in the \ref{P42} case) defines a B\"acklund transformation for equations P$_4^i$, since $u_{n-1}$ satisfies the same equations as $u_n$, up to the values of parameters.

As mentioned in Section \ref{s:mod}, the existence of substitutions between \ref{VL1} and \ref{VL2} suggests that there should be some transformation between equations \ref{P41} and \ref{P42} (which reduces to the identity transformation in the scalar case). However, its explicit form is still unknown.

Isomonodromic Lax pairs for the obtained equations are constructed according to the scheme from the previous section. We simplify the entries of the matrix $V_n=V^{(t_2)}_n+2xU_n+2V^{(\tau_1)}_n$ by use of the constraint dP$_1^i$ and arrive to the following representations.

\begin{proposition}
Equation \ref{dP11} and system (\ref{P41.pq}) for the variables $(u_{n-1},u_n)$ admit the Lax representations of the form (\ref{VL1.Llambda}) and (\ref{Ulambda}), respectively, where $\kappa(\lambda)=2\lambda$, the matrices $L_n$ and $U_n$ are of the form (\ref{VL1.LU}) and
\[
 V_n=\begin{pmatrix}
  \lambda^2+2x\lambda+\gamma_{n+1}+2\nu+1+\lambda u_n & 
  \lambda^2u_n -\lambda u_nu_{n-1}-\lambda\gamma_n\\
  -\lambda-2x-u_n-u_{n-1}& \gamma_n+2\nu+1-\lambda u_n 
 \end{pmatrix}.
\]
Equation \ref{dP12} and system (\ref{P42.pq}) for the variables $(u\tp_{n-1},u_n)$ admit the Lax representations of the form (\ref{VL2.Llambda}) and (\ref{Ulambda}), respectively, where $\kappa(\lambda)=2\lambda$, the matrices $L_n$ and $U_n$ are of the form (\ref{VL2.LU}) and
\begin{align*}
 V_n&=\frac{\lambda}{2}(\lambda+2x+2u_{n-1})\begin{pmatrix} 1 & 0 \\ 0 & -1 \end{pmatrix}\\
 &\qquad +\begin{pmatrix}
   0 & \lambda+2x+u_{n-1}+u\tp_n\\
  -\lambda^2u_{n-1}+\lambda u\tp_nu_{n-1}+\lambda\gamma_{n-1} & 
  1-2(-1)^n\varepsilon+[u\tp_n,u_{n-1}] 
 \end{pmatrix}.
\end{align*}
\end{proposition}

\section{Non-Abelian analogs of equations \texorpdfstring{dP$_{34}$}{dP34}, \texorpdfstring{P$_5$}{P5} and \texorpdfstring{P$_3$}{P3}}\label{s:constr2}

For the \ref{VL1} equation, the constraint (\ref{constr2}) is equivalent to equation
\begin{gather*}
 x(u_{n+2}u_{n+1}u_n+u^2_{n+1}u_n+u_{n+1}u^2_n-u^2_nu_{n-1}-u_nu^2_{n-1}-u_nu_{n-1}u_{n-2})\\
  +\bigl(n-\nu+\tfrac{3}{2}\bigr)u_{n+1}u_n+u^2_n-\bigl(n-\nu-\tfrac{3}{2}\bigr)u_nu_{n-1}+[s_n,u_n]\\
  -2\mu x(u_{n+1}u_n-u_nu_{n-1})-2\mu u_n =0,\qquad s_n-s_{n-1}=u_n.
\end{gather*}
Like in the previous section, it is represented in the form $F_{n+1}u_n-u_nF_{n-1}=0$, where
\[
 F_n=x(u_{n+1}u_n+u^2_n+u_nu_{n-1}-2\mu u_n) +\bigl(n-\nu-\tfrac{1}{2}\bigr)u_n+s_n -\mu n-(-1)^n\varepsilon-\varepsilon_0,
\]
with arbitrary constants $\varepsilon,\varepsilon_0\in\mathbb C$, and one can check that $F_n$ satisfy the identity $F_{n,x}=(F_{n+1}-F_n)u_n+u_n(F_n-F_{n-1})$. Hence, equation $F_n=0$ defines a constraint consistent with \ref{VL1}. In contrast to \ref{dP11}, it contains the nonlocal variable $s_n$, even in the scalar case. To get rid of it, we replace this equation with $G_n=F_{n+1}-F_n=0$, which is equation (\ref{Gn1}) below. 

For the \ref{VL2} equation, the constraint (\ref{constr2}) is equivalent to
\begin{gather*}
 x\bigl(u\tp_{n+1}u_{n+2}u_n +(u\tp_{n+1})^2u_n +u_nu\tp_{n+1}u_n 
  -u_nu\tp_{n-1}u_n -u_n(u\tp_{n-1})^2 -u_nu_{n-2}u\tp_{n-1}\bigr)\\
  +\bigl(n-\nu+\tfrac{3}{2}\bigr)u\tp_{n+1}u_n+u^2_n-\bigl(n-\nu-\tfrac{3}{2}\bigr)u_nu\tp_{n-1}
  -2\mu x(u\tp_{n+1}u_n-u_nu\tp_{n-1})-2\mu u_n =0.
\end{gather*}
This equation can be represented in the form $G_nu_n+u_nG\tp_{n-1}=0$ and it turns out that the equality $G_n=0$ (equation (\ref{Gn2}) below) is consistent with the derivation defined by \ref{VL2}. 

Thus, at this stage, we have replaced the original constraint with a 4-point difference equation, for both lattices.

\begin{proposition}
For any constants $\mu,\nu,\varepsilon\in\mathbb C$, the \ref{VL1} equation is consistent with the constraint
\begin{equation}\label{Gn1}
\begin{gathered}[b]
 x(u_{n+2}u_{n+1}+u^2_{n+1}-u^2_n-u_nu_{n-1})-(2\mu x-n+\nu-\tfrac{3}{2})u_{n+1}\\
 +(2\mu x-n+\nu+\tfrac{1}{2})u_n-\mu+2(-1)^n\varepsilon=0,
\end{gathered}
\end{equation}
and the \ref{VL2} equation is consistent with the constraint
\begin{equation}\label{Gn2}
\begin{gathered}[b]
 x\bigl(u\tp_{n+1}u_{n+2}+(u\tp_{n+1})^2-u^2_n-u_nu\tp_{n-1}\bigr)-(2\mu x-n+\nu-\tfrac{3}{2})u\tp_{n+1}\\
 +(2\mu x-n+\nu+\tfrac{1}{2})u_n-\mu+2(-1)^n\varepsilon=0.
\end{gathered}
\end{equation}
\end{proposition}
\begin{proof}
Let $G_n$ be the left-hand side of (\ref{Gn1}). A direct computation proves that its derivative in virtue of \ref{VL1} satisfies the identity $G_{n,x}=G_{n+1}u_{n+1}+u_{n+1}G_n-G_nu_n-u_nG_{n-1}$, which implies that the equality $G_n=0$ defines a constraint for this lattice equation. Similarly, if $G_n$ is the left-hand side of (\ref{Gn2}) then differentiation due to \ref{VL2} leads to the identity $G_{n,x}=u\tp_{n+1}(G\tp_{n+1}+G_n)-u_n(G_n+G\tp_{n-1})$.
\end{proof}

To obtain the Painlev\'e equations, we have to reduce the order by one more unit. It is far from immediately clear whether this can be done at all. In the scalar case, it was shown \cite{Adler_Shabat_2019} that equation (\ref{Gn1}) admits the integrating factor $xu_{n+1}+xu_n+n-\nu+\tfrac{1}{2}$. After multiplying by it, the equation is reduced to the form $H_{n+1}-H_n=0$, where $H_n$ is a cubic polynomial in $u_n$. The constraint $H_n=\const$ takes the form of the discrete Painlev\'e equation dP$_{34}$ \cite{Grammaticos_Ramani_2014}
\begin{equation}\label{dP34}\tag*{\mbox{dP$_{34}$}}
 (z_{n+1}+z_n)(z_n+z_{n-1})= 4x\frac{\mu z^2_n+2(-1)^n\varepsilon z_n+\delta}{z_n-n+\nu}
\end{equation}
after the change
\begin{equation}\label{zu}
 z_n=2xu_n+n-\nu.
\end{equation}
Moreover, the $x$-evolution is governed by continuous Painlev\'e equations. If $\mu\ne0$ then the function $y(x)=1-4\mu x/(z_{n+1}(x)+z_n(x))$ satisfies the P$_5$ equation, for any $n$, and if $\mu=0$ then the function $y(\xi)=(z_{n+1}(x)+z_n(x))/(2\xi)$ with $x=\xi^2$ satisfies the P$_3$ equation.

Our goal is to generalize these results to the non-commutative case. It is convenient to work directly with variables (\ref{zu}), for which the \ref{VL1} and \ref{VL2} equations take, respectively, the form
\begin{equation}\label{VL1z}
 2xz_{n,x}=z_{n+1}(z_n-n+\nu)-(z_n-n+\nu)z_{n-1}
\end{equation}
and
\begin{equation}\label{VL2z}
 2xz_{n,x}=z\tp_{n+1}(z_n-n+\nu)-(z_n-n+\nu)z\tp_{n-1}.
\end{equation}
The non-Abelian analogs of \ref{dP34} are given in the following statement. In contrast to the scalar situation, the cases $\mu\ne0$ and $\mu=0$ must be considered separately. If all variables are scalars then each of four equations coincides with \ref{dP34} up to coefficients.

\begin{proposition}
For $\mu\ne0$, let $\sigma=\varepsilon/\mu$ and let $\omega\in\mathbb C$ be an arbitrary constant. Then equation (\ref{Gn1}) admits the partial first integral consistent with  (\ref{VL1z}):
\begin{equation}\label{dP341}\tag*{\mbox{dP$_{34}^1$}}
 (z_{n-1}+z_n)(z_n+(-1)^n\sigma+\omega)^{-1}(z_n+z_{n+1})=4\mu x(z_n-n+\nu)^{-1}(z_n+(-1)^n\sigma-\omega),
\end{equation}
and equation (\ref{Gn2}) admits the partial first integral consistent with (\ref{VL2z}):
\begin{equation}\label{dP342}\tag*{\mbox{dP$_{34}^2$}}
 (z\tp_{n-1}+z_n)(z_n+(-1)^n(\sigma-\omega))^{-1}(z_n+z\tp_{n+1})
  =4\mu x(z_n-n+\nu)^{-1}(z_n+(-1)^n(\sigma+\omega)).
\end{equation}
For $\mu=0$, let $\delta\in\mathbb C$ be an arbitrary constant, then equation (\ref{Gn1}) admits the alternating partial first integral consistent with (\ref{VL1z}):
\begin{equation}\label{dP3410}\tag*{\mbox{d$\widetilde{\rm P}_{34}^1$}}
\left\{\begin{array}{ll}
 (z_{n+1}+z_n)(z_n-n+\nu)(z_n+z_{n-1})=4x( 2\varepsilon z_n+\delta), & n=2k,\\
 (z_n+z_{n-1})(z_{n+1}+z_n)(z_n-n+\nu)=4x(-2\varepsilon z_n+\delta), & n=2k+1,
\end{array}\right.
\end{equation}
and equation (\ref{Gn2}) admits the partial first integral consistent with (\ref{VL2z}):
\begin{equation}\label{dP3420}\tag*{\mbox{d$\widetilde{\rm P}_{34}^2$}}
 (z\tp_{n+1}+z_n)(z_n-n+\nu)(z_n+z\tp_{n-1})=4x(2(-1)^n\varepsilon z_n+\delta).
\end{equation}
\end{proposition}

\begin{proof}
The statement can be proved by direct, but rather tedious computations. In the case $\mu\ne0$, a more conceptual proof can be obtained by use of representations (\ref{VL1.Llambda})--(\ref{Ulambda}). We have $\kappa(\lambda)=\lambda^2-2\mu\lambda$, therefore for $\lambda=2\mu$ we have the equation $V_{n+1}L_n=L_nV_n$ for the case of \ref{VL1} or $V_{n+1}L_n=-L_nV\tp_n$ for the case of \ref{VL2}, and, in both cases, $V_{n,x}=[U_n,V_n]$. From here it is not difficult to prove, in a general form, that vanishing of a quasi-determinant $\Delta_n$ of the matrix $V_n|_{\lambda=2\mu}$ defines a constraint which is consistent both with continuous and discrete dynamics. For instance, the equation
\[
  V_x=[U,V],\quad V=\begin{pmatrix} a & b\\ c & d\end{pmatrix},\quad U=\begin{pmatrix} p & q\\ r & s\end{pmatrix},
\]
implies the following relation for $\Delta=b-ac^{-1}d$:
\[
 \Delta_x=(p-ac^{-1}r)\Delta-\Delta(s-rc^{-1}d) \quad\Rightarrow\quad \Delta_x|_{\Delta=0}=0;
\]
in a similar way, the discrete equations imply the relations of the form $\Delta_{n+1}=f_n\Delta_ng_n$ or $\Delta_{n+1}=f_n\Delta\tp_ng_n$ with some factors $f_n,g_n$. It is easy to verify that equations \ref{dP341} and \ref{dP342} are nothing but the equality $\Delta_n=0$ for the respective matrices $V_n|_{\lambda=2\mu}$, simplified modulo equations (\ref{Gn1}) or (\ref{Gn2}) and the change (\ref{zu}).
\end{proof}

In conclusion, we write down the systems arising from the lattices (\ref{VL1z}) and (\ref{VL2z}) under the above constraints. Each of these sytems admits the isomonodromic Lax pair (\ref{Ulambda}) with the matrices constructed from the matrices from Section \ref{s:zcr}.

For the lattice (\ref{VL1z}) and the constraint \ref{dP341}, the variables $q=z_n$ and $p=z_n+z_{n+1}$ satisfy, for any $n$, the system of the form
\begin{equation}\label{P51.pq}
\left\{\begin{array}{l}
 2xq_x=p(q-n+\nu)-4\mu x(q+\alpha)p^{-1}(q+\beta),\\
 2xp_x=pq+qp+p-p^2+4\mu x(p-2q-\alpha-\beta),
\end{array}\right.
\end{equation}
where $\alpha=(-1)^n\sigma-\omega$ and $\beta=(-1)^n\sigma+\omega$. For the lattice (\ref{VL2z}) and the constraint \ref{dP342}, the variables $q=z_n$ and $p=z_n+z\tp_{n+1}$ satisfy an almost identical system (cf. with the pair (\ref{P41.pq}) and (\ref{P42.pq}))
\begin{equation}\label{P52.pq}
\left\{\begin{array}{l}
 2xq_x=p(q-n+\nu)-4\mu x(q+\alpha)p^{-1}(q+\beta),\\
 2xp_x=2pq+p-p^2+4\mu x(p-2q-\alpha-\beta),
\end{array}\right.
\end{equation}
where $\alpha=(-1)^n(\sigma+\omega)$ and $\beta=(-1)^n(\sigma-\omega)$. Notice, that the system (\ref{P52.pq}) reduces to one rational equation of second order, by solving its second equation with respect to $q$ and substituting the result into the first equation (recall that in the scalar case the P$_5$ equation is satisfied by the variable $y=1-4\mu xp^{-1}$). However, this is hardly advisable, since the resulting formulas are rather cumbersome. Whether it is possible to eliminate $q$ in the system (\ref{P51.pq}) is not obvious.

In a similar way, for the lattice (\ref{VL1z}) and the constraint \ref{dP3410}, the variables $q=z_n$ and $p=z_n+z_{n+1}$ with even $n$ satisfy the system
\begin{equation}\label{P31.pq}
\left\{\begin{array}{l}
 2xq_x=p(q-n+\nu)-4xp^{-1}(2\varepsilon q+\delta),\\
 2xp_x=pq+qp+p-p^2-8\varepsilon x
\end{array}\right.
\end{equation}
(of course, a similar system can be derived also for odd $n$), and for the lattice (\ref{VL2z}) and the constraint \ref{dP3420}, the variables $q=z_n$ and $p=z_n+z\tp_{n+1}$ satisfy, for any $n$, the system
\begin{equation}\label{P32.pq}
\left\{\begin{array}{l}
 2xq_x=p(q-n+\nu)-4xp^{-1}(2(-1)^n\varepsilon q+\delta),\\
 2xp_x=2pq+p-p^2-8(-1)^n\varepsilon x.
\end{array}\right.
\end{equation}
Recall that in the scalar case the P$_3$ equation is obtained for the variable $y=p/(2\xi)$, after the change $x=\xi^2$.

\subsubsection*{Acknowledgements}

I am grateful to V.V.~Sokolov for many stimulating discussions. This research was carried out under the State Assignment 0033-2019-0004 (Quantum field theory) of the Ministry of Science and Higher Education of the Russian Federation.


\end{document}